\documentclass{article}
\usepackage{Definitions}
\usepackage{url}
\usepackage[most]{tcolorbox}
\usepackage{fullpage}
\usepackage{hyperref}
\usepackage[normalem]{ulem}
\usepackage{cancel}
\usepackage{authblk}
\usepackage[numbers,sort&compress]{natbib}
\usepackage{xspace}
\usepackage{bbm}
\usepackage{algorithm}
\usepackage[noend]{algpseudocode}
\usepackage{amssymb}
\usepackage{subfig}
\usepackage{lineno}

\newcommand{\covid}{CO\-VID-19\xspace}
\newcommand{\rnumber}{r\xspace}
\newcommand{\poisson}{\text{Pois}\xspace}
\newcommand{\nbin}{\text{NBin}\xspace}
\newcommand*{\given}{\ensuremath{\,|\,}}

\newcommand\indicator[1]{{\mathbbm{1}[#1]}}

\title{Group Testing under Superspreading Dynamics}

\author[$\mathsection$]{Stratis Tsirtsis}
\author[$*$]{Abir De}
\author[$\dagger$]{Lars Lorch}
\author[$\mathsection$]{Manuel Gomez Rodriguez}

\affil[$\mathsection$]{Max Planck Institute for Software Systems,  
\href{mailto:manuelgr@mpi-sws.org}{\{stsirtsis, manuelgr\}@mpi-sws.org}}
\affil[$*$]{IIT Bombay\\ \href{mailto:abir@cse.iitb.ac.in}{abir@cse.iitb.ac.in}}
\affil[$\dagger$]{ETH Zurich\\
\href{mailto:llorch@student.ethz.ch}{llorch@student.ethz.ch}
}

\date{}

\graphicspath{{./fig/}}

\begin{document}

\maketitle

\begin{abstract}
%
Testing is recommended for all close contacts of confirmed COVID-19 patients. 
%
However, existing group testing methods are oblivious to the circumstances of contagion 
provided by contact tracing.
Here, we build upon a well-known semi-adaptive pool testing method, Dorfman'{}s
method with imperfect tests, and derive a simple group testing method based on 
dynamic programming that is specifically designed to use the information 
provided by contact tracing.
Experiments using a variety of reproduction numbers and dispersion levels, including
those estimated in the context of the COVID-19 pandemic, show that the pools found
using our method result in a significantly lower number of tests than those found using standard Dorfman'{}s method, especially when the number of contacts of an infected individual is small.
Moreover, our results show that our method can be more beneficial when the secondary infections are highly overdispersed.
%
\end{abstract}

\section*{Introduction}
\label{sec:introduction}
As countries around the world learn to live with \covid, the use of testing, contact tracing and 
isolation (TTI) has been proven to be as important as social distancing for containing the spread
of the disease~\cite{contreras2020challenges}.
However, as the infection levels grow, TTI reaches a tipping point and its effectiveness 
quickly degrades as the health authorities lack resources to trace and test all contacts 
of diagnosed individuals~\cite{linden2020foreshadow}.
In this context, there has been a flurry of interest on the use of group testing---testing 
groups of samples simultaneously---to scale up testing under limited resources. 

The literature on group testing methods has a rich history, starting with the seminal work by Dorfman \cite{dorfman1943detection, graff1972group, aprahamian2019optimal}.
However, existing methods~\cite{aldridge2019group,du2000combinatorial}, including those developed and used in the context of the \mbox{\covid} pandemic~\cite{agoti2020pooled,eberhardt2020multi,sunjaya2020pooled,mutesa2020pooled,deckert2020simulation,noriega2020increasing, ben2020large, hanel2020boosting,zhu2020noisy}, are oblivious to the circumstances of contagion provided by contact tracing and assume statistical independence of the samples to be tested. 
This assumption can be seemingly justified by classical epidemiological models where the number 
of infections caused by a single individual follows a Poisson distribution. 
However, in \covid, there is growing evidence suggesting that the number of secondary
infections caused by a single individual is \emph{overdispersed}---most individuals do not infect anyone but a few \emph{superspreaders} infect many in infection hotspots~\cite{adam2020clustering,carehomes,endo2020estimating,lau2020characterizing,frieden20identifying,athreya2020effective, lorch2020spatiotemporal}.
Overdispersion 
has been also observed in MERS and SARS~\cite{oh2015middle,lloyd2005superspreading,stein2011super}.
In this work, our goal is to develop group testing methods that  
are specifically designed to use the information provided by contact tracing and are effective in the presence of overdispersion.

More specifically, we build upon a well-known semi-adaptive pool testing method, Dorfman'{}s method 
with imperfect tests~\cite{hwang1975generalized, aprahamian2019optimal}. 
In Dorfman'{}s method, samples from multiple individuals are first pooled together and 
evaluated using a single test. 
If a pooled sample is negative, all individuals in the pooled sample are deemed negative. If the pooled sample is positive, each individual sample from the pool is then tested independently.
However, rather than modeling the probability that each individual sample is positive using an independent 
Bernoulli distribution as in Dorfman'{}s method, 
we assume that: 
%
(i) the samples to be tested are all the contacts of a diagnosed individual during their infectious
period, who are identified using contact tracing, and
%
(ii) the number of true positive samples, \ie, secondary infections by the diagnosed individual, 
follows an overdispersed negative binomial distribution~\cite{athreya2020effective, lorch2020spatiotemporal}.
%
Given any arbitrary set of pools, we then compute the average number of tests and the expected number of 
false negatives and false positives under our model. 
Finally, we introduce a dynamic programming algorithm to efficiently find the set of 
pools that optimally trade off the average number of tests, false negatives and false 
positives in polynomial time.
Experiments using a variety of reproduction numbers and dispersion levels in secondary infections, including
those observed for \covid, show that the pools found
using our method result in a significantly lower average number of tests than those found using standard Dorfman'{}s method, especially when the number of contacts of an infected individual is small.
Moreover, our results show that our method can be particularly beneficial when the number of secondary infections caused by an infectious individual is highly overdispersed.
%


\section*{Methods}
\label{sec:model}
\subsection*{Modeling overdispersion of infected contacts}
Previous work have mostly built on the assumption that the number of infections $X$ caused by a single individual follows a Poisson distribution with mean $\rnumber$, so $X \sim \poisson(\rnumber)$, where
$\rnumber$ is often called the effective reproduction number.
However, having equal mean and variance, the Poisson is unable to capture settings where the number of cases to be tested for
exhibits higher variance.
Following recent work in the context of \covid~\cite{endo2020estimating,athreya2020effective}, 
we instead model $X$ using a generalized negative binomial distribution.
For a (standard) negative binomial distribution, $X \sim \nbin(k,p)$ can be interpreted as the 
number of successes before the $k$-th failure in a sequence of Bernoulli trials with success probability $p$. 
For a generalized negative binomial distribution, $k > 0$ can take real values and the probability mass
function is given by
\begin{align*} 
    P(X = n) 
    &= \frac{\Gamma(n + k)}{\Gamma(k) n!} p^n (1-p)^k,
\end{align*}
where $k$ is called the dispersion parameter and parameterizes higher variance of the distribution 
for small $k$. 
%
%
Since $\EE[X] = kp/(1-p)$, we assume in this work that the number of secondary infections $X$ is distributed as  $X \sim \nbin\left ( k, p \right)$ with $p = \rnumber / (k + \rnumber)$,
hence parameterizing $X$ via its mean $\EE[X] = r$ and dispersion parameter $k$. 
Under this parameterization, $\text{Var}[X] = r(1 + r/k)$, which is greater than the variance of the Poisson $r$ for $k < \infty$.
For $k \rightarrow \infty$, the sequence of random variables $X_k \sim \nbin(k, r/(k + r))$ converges in distribution to $X \sim \poisson(r)$, making the negative binomial a suitable generalization of the Poisson for modeling secondary infections.

Assuming we test all contacts of a diagnosed individual during their infectious period, we have prior information about the maximum number of possible infections $N$ in practice.
In this setting, we will write
\begin{align}
    q_{\rnumber, k, N}(n) := P(X = n \given X \leq N) \quad\text{when} \quad 
    X \sim \nbin
    \left ( k, \rnumber / (k + \rnumber) \right ). \label{eq:q}
\end{align}
Here, note that $P(X = n \given X \leq N) = P(X = n) / P(X \leq N)$ if $n \leq N$ and 0 otherwise.

\subsection*{Pooling contacts of a positively diagnosed individual}
In this context, our goal is to identify infected individuals among all contacts $\Ncal$ of a positively diagnosed individual via testing, where $|\Ncal| = N$.
For events concerning each individual $j \in \Ncal$, we define the following indicator random variable:
\begin{align*} 
    I_j &= \indicator{\text{individual $j$ is infected}}
\end{align*}
In addition, for each pool of individuals $\Scal \subseteq \Ncal$, 
we define the number of infected in $\Scal$ as $I(\Scal) := \sum_{j \in \Scal} I_j$.
%
%
%
Following our assumption on the distribution of the number of secondary infections, we define
\begin{align*}
    P(I(\Ncal) = n) = q_{\rnumber, k, N}(n)
\end{align*}
Let $T(\Scal) = \indicator{\text{test of pool $\Scal$ is positive}}$.
%
To account for the sensitivity $s_e$ (\ie, true positive probability) and
specificity $s_p$ (\ie, true negative probability) of tests, 
we parameterize the conditional probabilities as
\begin{align*} 
    P(T(\Scal) = 1 \given I(\Scal)>0) &= s_e\\
    P(T(\Scal) = 0 \given I(\Scal) = 0) &= s_p
\end{align*}
%
In the above, following the literature on the group testing \cite{aprahamian2019optimal}, we assume that the sensitivity of pool tests is independent of the exact number of infected individuals in the pool.
Moreover, while dilution has been shown to often be negligible~\cite{yelin2020evaluation}, 
the effect can easily be addressed by making the conditional $T(\Scal) \given I(\Scal)$, 
dependent on the corresponding pool size $|\Scal|$.

\subsection*{Dorfman testing under overdispersion of infected contacts}
Dorfman testing proceeds by pooling individuals into non-overlapping partitions of $\Ncal$ 
and first testing the combined samples of each pool using a single test.
Every member of a pool is marked as negative if their combined sample is negative. 
If a combined sample of a pool is positive, each individual of the pool is subsequently tested individually to determine who exactly is marked positive in the pool.

Let $D^{\Scal}_j$ denote the indicator random variable for the event that individual $j$ is marked as infected in pool $\Scal \subseteq \Ncal : |\Scal| > 1$ after Dorfman testing.
Then, $D^{\Scal}_j$ can be expressed as
\begin{align*} 
    D^{\Scal}_j &= \indicator{T(\Scal) = 1 \cap T(\{j\}) = 1}, 
\end{align*}
%
\ie, taking value 1 if and only if the combined sample of pool $\Scal$ 
is first tested positive \emph{and} the sample of individual $j$ is tested positive in the second step.
%
%
In the simple case of $|\Scal|=1$, we have $D^{\Scal}_j = T(\{j\})$.
%

\paragraph{Expected number of tests}
Let $K(\Scal)$ be the number of tests performed when testing pool $\Scal$ as described above.
Then, the expected number of tests  $\EE[K(\Scal)]$ due to a pool
$\Scal$ is:
\begin{align*} 
\EE[K(\Scal)] &= 
\begin{cases}
1 + f(\Scal) & |\Scal| > 1 \\
1 & |\Scal| = 1
\end{cases}
\end{align*}
where $f(\Scal)$ is given by
\begin{align*}
    f(\Scal) 
    &= |S| \Big [ 1 - P\big(T(\Scal) = 0\big)\Big]\\
    &= |S| \Bigg [ 1 - 
        \sum_{s=0}^{|\Scal|} 
        P\big(T(\Scal) = 0 \given I(\Scal) = s\big) 
        P\big(I(\Scal) = s\big)
    \Bigg]\\
    &= |S| \Bigg [ 1 - 
        \sum_{s=1}^{|\Scal|}
        P\big(T(\Scal) = 0 \given I(\Scal) = s\big) 
        \sum_{n=s}^N 
            P\big(I(\Scal) = s \given I(\Ncal) = n\big)  
            P\big(I(\Ncal) = n\big)\\
    &\quad\quad\quad -  P\big(T(\Scal) = 0 \given I(\Scal) = 0\big) 
        \sum_{n=0}^N 
            P\big(I(\Scal) = 0 \given I(\Ncal) = n\big)  
            P\big(I(\Ncal) = n\big)
    \Bigg]\\
    &= |S| \Bigg [ 1 - 
        \sum_{s=1}^{|\Scal|}
        (1 - s_e) 
        \sum_{n=s}^N 
            \frac
                {\binom{n}{s}\binom{N-n}{|\Scal| - s}}
                {\binom{N}{|\Scal|}}
            q_{\rnumber, k, N}(n)
        -  s_p 
        \sum_{n=0}^N 
            \frac
                {\binom{N-n}{|\Scal|}}
                {\binom{N}{|\Scal|}}
            q_{\rnumber, k, N}(n)
    \Bigg]
\end{align*}
where the last step follows from the fact that 
$I(\Scal) \given I(\Ncal) = n  \sim \text{HGeom}(N,n,|\Scal|)$, our assumption about $\PP (I(\Ncal) = n)$ and our assumptions about $T(\Scal)$.

\paragraph{Expected number of false negatives}
To compute the number of false negatives, we distinguish between two cases.
If $|\Scal|=1$, \ie, the pool consists of only one person, there is no distinction between a group test and an individual test. Therefore, a false negative can occur only if the person is infected and the test turns out negative.
Thus, 
\begin{equation*}
    \EE[FN(\Scal)] = (1-s_e)P(I(\Scal)=1) =
    (1-s_e)\sum_{n=1}^N 
    \frac{n}{N}q_{\rnumber, k, N}(n)
\end{equation*}
If $|\Scal|>1$, a pooled test is performed and, if it turns out positive, individual tests are performed afterwards. The expected number of false negatives in this case is
\begin{equation*}
    \EE[FN(\Scal)] = \sum_{s = 1}^{|\Scal|} s \, P(T(\Scal)=0 \given I(\Scal) = s) P(I(\Scal) = s) + \sum_{s = 1}^{|\Scal|} P(T(\Scal)=1 \given I(\Scal) = s) P(I(\Scal) = s) \, s \, (1-s_e)
\end{equation*}
where the first term corresponds to the case where the group test outcome is falsely negative and the second term corresponds to the case where the group test outcome is truly positive and the individual tests are falsely negative.
Using our assumptions about the individual probabilities, we can rewrite the above expression as
\begin{align*}
    \EE[FN(\Scal)] = \sum_{s = 1}^{|\Scal|} s (1-s_e) P(I(\Scal) = s) + \sum_{s = 1}^{|\Scal|}  s (1-s_e) s_e P(I(\Scal) = s)
    = \sum_{s = 1}^{|\Scal|} s (1-s_e^2) \left[ \sum_{n=s}^N 
            \frac
                {\binom{n}{s}\binom{N-n}{|\Scal| - s}}
                {\binom{N}{|\Scal|}}
            q_{\rnumber, k, N}(n) \right]
\end{align*}

\paragraph{Expected number of false positives}  We likewise distinguish between the two cases for computing the number of false positives.
Again, if $|\Scal|=1$, there is no distinction between a group test and an individual test. Therefore, a false positive can occur only if the person is not infected and the test turns out positive.
Thus,
\begin{equation*}
    \EE[FP(\Scal)] = (1-s_p)P(I(\Scal)=0) =
    (1-s_p)\sum_{n=0}^{N-1} 
    \frac{N-n}{N}q_{\rnumber, k, N}(n)
\end{equation*}
If $|\Scal|>1$, a group test is performed and, after a positive result, individual tests are performed subsequently. 
Truly negative subjects are falsely classified as 
positive if the corresponding group test outcome is positive and the 
subject'{}s subsequent individual test outcome is positive, \ie,
\begin{equation*}
\EE[FP(\Scal)] = \sum_{s = 0}^{|S|-1} P(T(\Scal) = 1 \given I(\Scal) = s) P(I(\Scal) = s) (|\Scal|-s) (1-s_p) 
\end{equation*}
Finally, under our assumptions about the individual probabilities, we rewrite the above expression as
\begin{align*}
    \EE[FP(\Scal)] &=  (1-s_p) P(I(\Scal) = 0) |\Scal| (1-s_p) + \sum_{s = 1}^{|S|-1} s_e P(I(\Scal) = s) (|\Scal|-s) (1-s_p)\\
	&= (1-s_p)^2 |\Scal|
	\left[ \sum_{n=0}^N 
            \frac
                {\binom{N-n}{|\Scal|}}
                {\binom{N}{|\Scal|}}
            q_{\rnumber, k, N}(n) \right]
     +
    \sum_{s=1}^{|\Scal|-1} s_e (|\Scal|-s)(1-s_p) 
    \left[ \sum_{n=s}^N 
            \frac
                {\binom{n}{s}\binom{N-n}{|\Scal| - s}}
                {\binom{N}{|\Scal|}}
            q_{\rnumber, k, N}(n) \right]
\end{align*}

\subsection*{Finding the optimal pool sizes}
\label{sec:algorithm}

First, we note that the expected number of tests, false negatives and false
positives only depend on the pool size. Therefore, overloading notation, for
a number of contacts $|\Ncal| = N$ and pool of size $|\Scal| = s$, we will write $E[K(s)]$, $\EE[FN(s)]$ and $\EE[FP(s)]$.


Our goal is to find the sizes $\{s_i\}$ of the optimal sets of pools that optimally trade off the expected number of tests, false negatives and false positives~\cite{aprahamian2019optimal}:
\begin{equation*}
    \underset{\{s_i\}}{\text{minimize}} \quad \sum_i g(s_i) \quad \text{subject to} \quad \sum_i s_i = N
\end{equation*}
with
\begin{equation*}
    g(s_i) =  \EE[K(s_i)] + \lambda_1 \EE[FN(s_i)] + \lambda_2 \EE[FP(s_i)],
\end{equation*}
where $\lambda_1$ and $\lambda_2$ are given non-negative parameters, penalizing the numbers of false negatives and false positives.
\begin{algorithm}[t]
  \renewcommand{\algorithmicrequire}{\textbf{Input:}}
  \caption{Find the sizes of the optimal set of pools under overdispersion of infected contacts} \label{alg:dp}
    \begin{algorithmic}[1]
    \Require Number of secondary contacts $N$, sensitivity $s_e$, specificity $s_p$, parameters $r$, $k$, $\lambda_1$ and $\lambda_2$
    \State $\Scal_{0} \gets \emptyset$
    \State $h(0) \gets 0$
    \For{$k \in \{1, \ldots, N\}$}
        \State $g(k) \gets \textsc{ComputeObjective}(k, N, s_e, s_p, r, k, \lambda_1, \lambda_2)$ 
    \EndFor 
    \For{$n \in \{1, \ldots, N\}$}
        \State $h(n) \gets \min_{1 \leq j \leq n} \left[g(j) + h(n-j)\right]$
        \State $s \gets \argmin_{1 \leq j \leq n} \left[g(j) + h(n-j)\right]$
        \State $\Scal_n = \Scal_{n-s} \cup \{s\}$
    \EndFor 
    \State \Return $\Scal_N$
    \end{algorithmic}
\end{algorithm}

Perhaps surprisingly, we can solve the above problem in polynomial time using a simple dynamic programming procedure. 
More specifically, define the following recursive functions:
%
\begin{align*}
    h(n) &= \min_{1 \leq j \leq n} [g(j) + h(n-j)] \\
    \Scal_n &= \Scal_{n-s} \cup \{s\},
\end{align*}
where $s = \argmin_{1 \leq j \leq n} [g(j) + h(n-j)]$.
%
%
Interpreting $n$ as the number of individuals not yet assigned to a pool, using the two recursive functions, the (sizes of the) optimal set of pools can be recovered by computing $h(n)$ in increasing order of $n$.
Algorithm~\ref{alg:dp} summarizes the overall procedure, where the function $\textsc{ComputeObjective}(\cdot)$ precomputes the function 
$g(k)$ for each $k \in \{1, \ldots, N\}$.
%
%
More formally, we arrive at the following proposition:
\begin{proposition}
Given $N$ contacts, the set of pool sizes $\Scal_N$ returned by Algorithm~\ref{alg:dp} 
are optimal.
\end{proposition}
%

\begin{proof}
We will prove this proposition by induction.
In the base case, where $n=1$, it is easy to see that the optimal solution is $\Scal^*_1 = \{1\}$ \ie, it consists of one group with size $1$ and the minimum of the objective value is $OPT_1 = g(1)$, while the recursive functions trivially find the optimal solution since $h(1) = g(1)$.
For $n>1$ contacts, the inductive hypothesis is that the values $h(i)$ and sets $\Scal_i$ recovered by Algorithm~\ref{alg:dp} for all $i<n$ are optimal.
Let $\Scal^*_n$ and $OPT_n$ be the optimal set of pool sizes and the respective value of the objective function for $n$ contacts.

Suppose, for the sake of contradiction, that $OPT_n<h(n)$ \ie, the solution computed using the recursive functions for $n$ contacts is {\em suboptimal}.
Let $\Scal^*_n = \{s^*_1, s^*_2, \ldots, s^*_l\}$. Then, we get:
\begin{equation*}
\sum_{i=1}^l g(s^*_i) < g(s) + h(n-s) \Rightarrow 
\sum_{i=1}^l g(s^*_i) < g(s^*_1) + h(n-s^*_1) \Rightarrow  
\sum_{i=2}^l g(s^*_i) < OPT_{n-s^*_1},
\end{equation*}
where the first step is based on the fact that $g(s) + h(n-s) \leq g(j) + h(n-j)$ for all $j: 1\leq j\leq n$ and the second step is based on the inductive hypothesis.
Since $\sum_{i=2}^l s^*_i = n-s^*_1$, the final inequality implies that, having $n-s^*_1$ contacts, the set of pool sizes $\{s^*_2, \ldots, s^*_l\}$ is strictly better than the optimal one which is clearly a contradiction.
Therefore, the values $h(n)$ and sets of pool sizes $\Scal_n$ given by the recursive functions are optimal for all $n: 1\leq n \leq N$.

\end{proof}

\section*{Results and Discussion}
\label{sec:results}
%
%
%

\begin{figure}[t]
\centering
\subfloat[Number of tests per contact]{\includegraphics[width=0.32\textwidth]{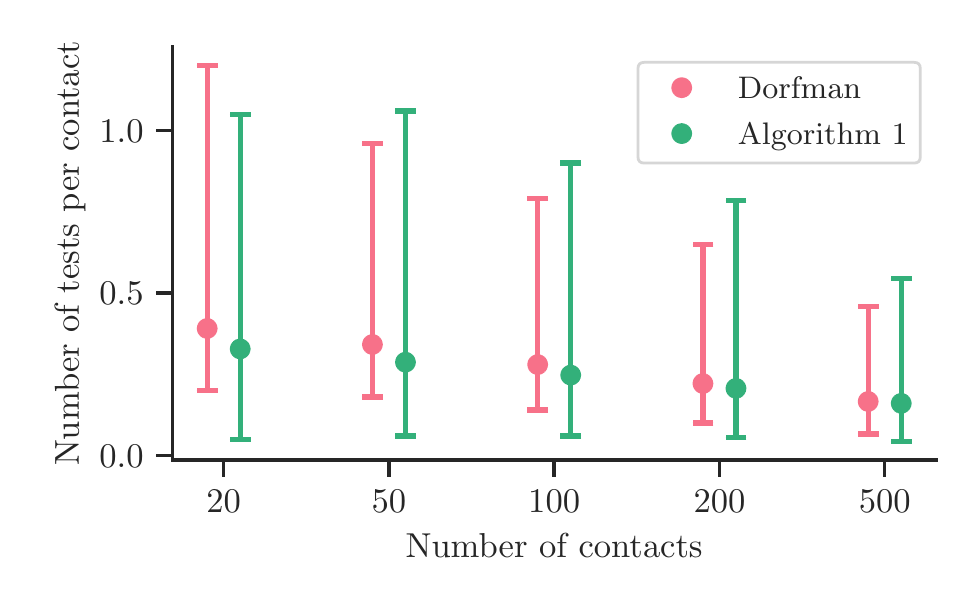}} 
\subfloat[Pool size]{\includegraphics[width=0.32\textwidth]{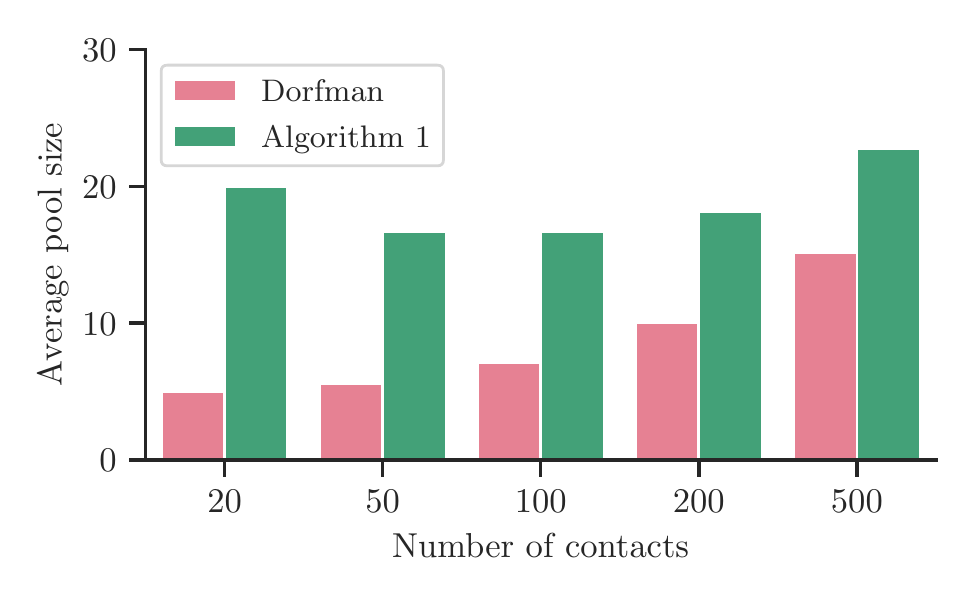}}
\subfloat[Percentage of tests saved]{\includegraphics[width=0.32\textwidth]{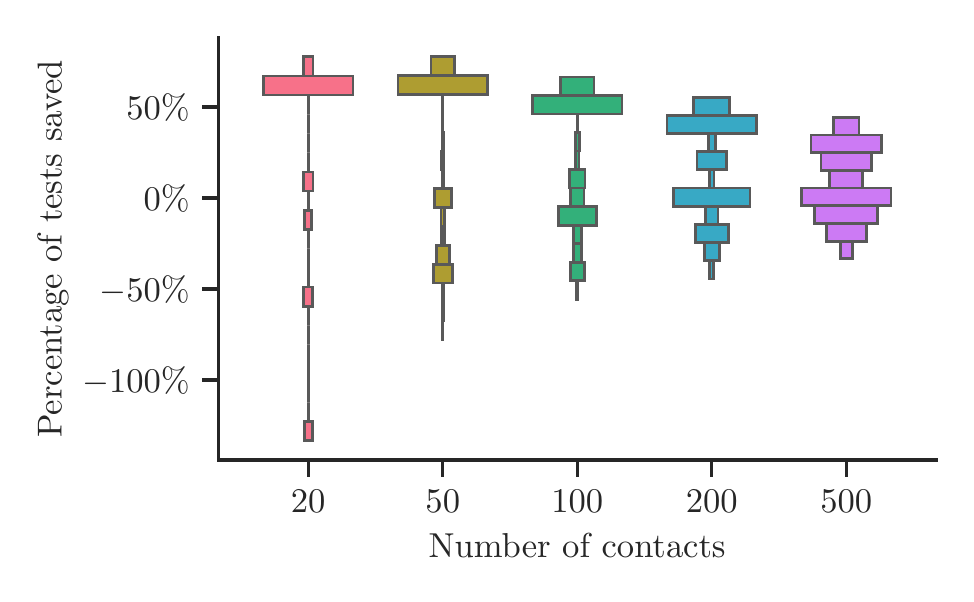}}
\caption{Performance of our method (Algorithm ~\ref{alg:dp}) and Dorfman'{}s method for various 
values of the number of contacts of a diagnosed individual during their infectious period.
Panel (a) shows the number of tests per contact, where the dots represent the average value and 
the error bars cover $90\%$ of the observations.
Panel (b) shows the average pool size.
Panel (c) shows the empirical distribution of the percentage of tests saved by using our method instead of Dorfman'{}s method,
where we exclude the highest and lowest $5\%$ of observations.
In all panels, we set the sensitivity and specificity to $s_e=s_p = 0.95$,
and sample the number of secondary infections from a truncated negative binomial distribution with mean 
$r = 2.5$ and dispersion parameter $k = 0.1$~\cite{endo2020estimating, kucharski2020effectiveness, lau2020characterizing, hasan2020superspreading}. 
For each combination of method and parameter values, the averages and quantiles in all panels are estimated using $100{,}000$ samples.
%
%
} \label{fig:performance-vs-N}
\end{figure}

%
We perform simulations to evaluate Algorithm~\ref{alg:dp} against Dorfman{}'s method in its ability to optimally trade off resources and false test outcomes in the presence of overdispersed distributions of secondary infections.
To generate the infection states for each contact, we first fix a number of contacts $N$ and sample the number of secondary infections $n \sim q_{r, k, N}(n)$, where $q_{r, k, N}(n)$ is a truncated negative binomial distribution as defined in Eq.~\ref{eq:q}.
Then, we select $n$ of the $N$ contacts at random and set their status to infected.
To implement Dorfman'{}s method, we use a variation of Algorithm~\ref{alg:dp} in which the expected numbers of tests, false negatives and false positives are computed assuming an independent individual probability of infection $p = \EE_{q_{r, k, N}}[n]/N$, using the formulas derived by Abrahamian et al.~\cite{aprahamian2019optimal}. 
%
%

%
We first compare the performance of our method and Dorfman'{}s method at finding the pools that minimize the number of tests (i.e., $\lambda_1 = \lambda_2 = 0$) for fixed values of the reproductive number $r$ and dispersion parameter $k$ matching estimations done during the early phase of the COVID-19 pandemic.
Figure~\ref{fig:performance-vs-N} summarizes the results with respect to the number of contacts $N$ of 
the diagnosed individuals. 
The results show that our method achieves a lower average number of tests across all values of $N$, with its competitive advantage being greater when the number of contacts is small and less apparent as the number of contacts is increasing.
Moreover, Dorfman'{} method chooses pool sizes that increase with the number of contacts while the ones chosen by our method remain relatively constant. 
This leads to significant differences between the distributions of the number of tests performed under the two methods. 
For example, as shown in Figure~\ref{fig:performance-vs-N}(c), when the number of contacts is $N=20$, our method is most likely to perform about $50\%$ less tests than Dorfman'{}s.
However, due to the more conservative pool sizes given by Dorfman'{}s method, there is a small probability that our method ends up performing more tests, sometimes even double the amount.
%
%
%

\begin{figure}[t]
\centering
\subfloat[$N = 20$]{\includegraphics[width=0.32\textwidth]{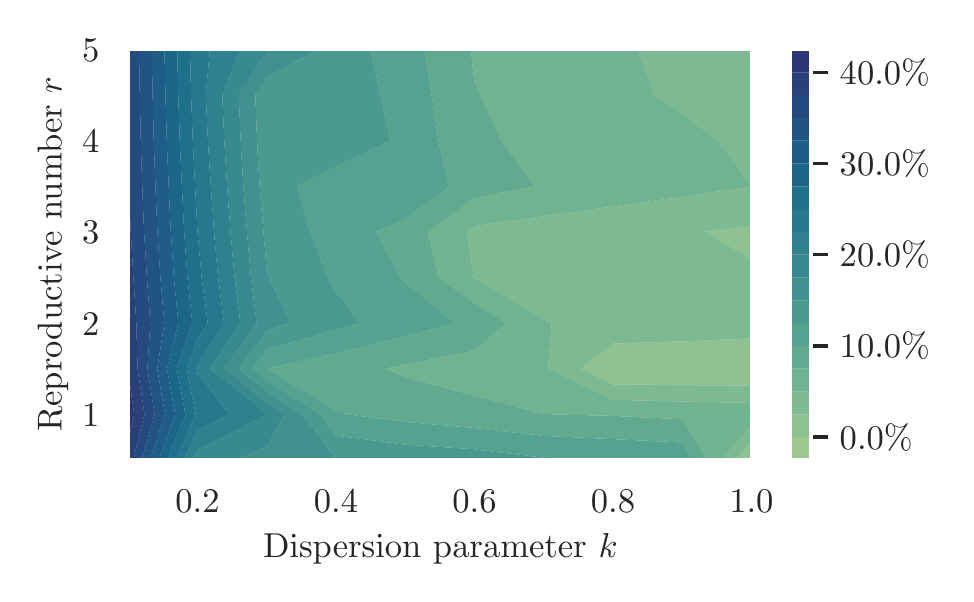}}
\subfloat[$N = 100$]{\includegraphics[width=0.32\textwidth]{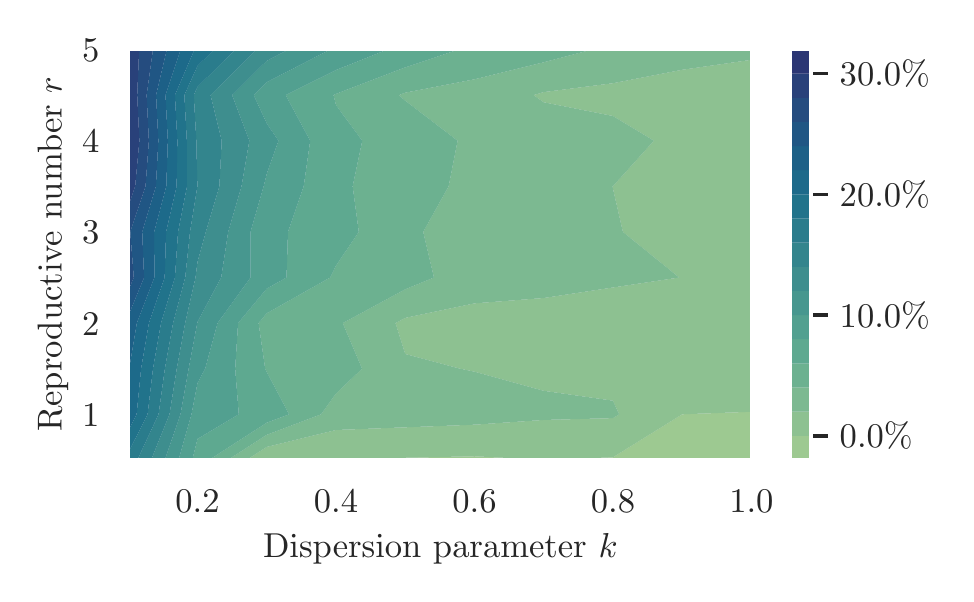}}
\subfloat[$N = 200$]{\includegraphics[width=0.32\textwidth]{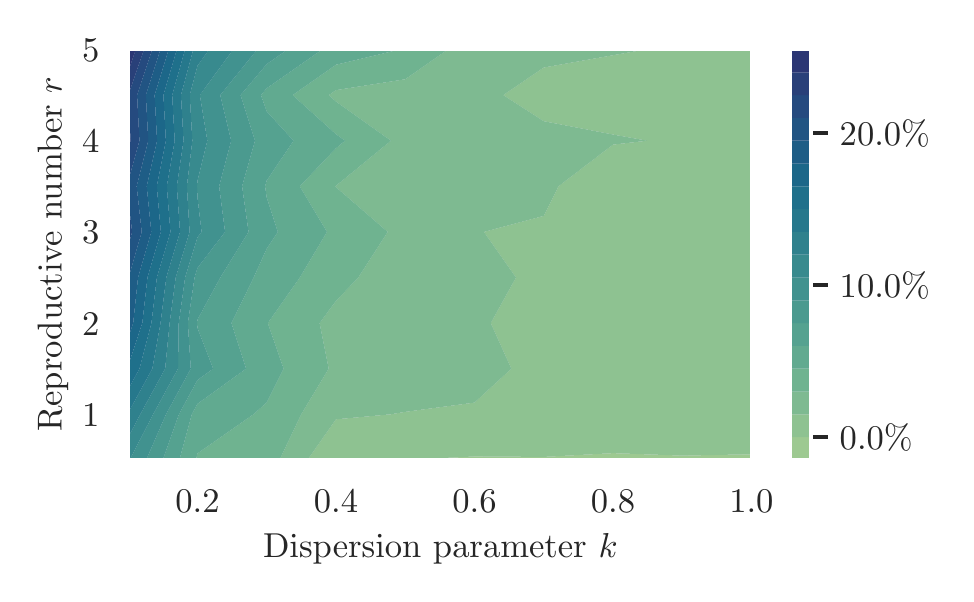}}
\caption{Performance of our method (Algorithm~\ref{alg:dp}) and Dorfman'{}s method for different values 
of the reproductive number $r$ and dispersion parameter $k$.
Each panel shows the average percentage of tests saved by using our method instead of Dorfman'{}s method.
Here, we set the sensitivity and specificity to $s_e =s_p=0.95$ and, in each experiment, we estimate the average using $100{,}000$ samples.} \label{fig:performance-vs-R-k}
\end{figure}

\begin{figure}[h!]
\centering
\subfloat[$s_e=s_p=0.75$]{\includegraphics[width=0.32\linewidth]{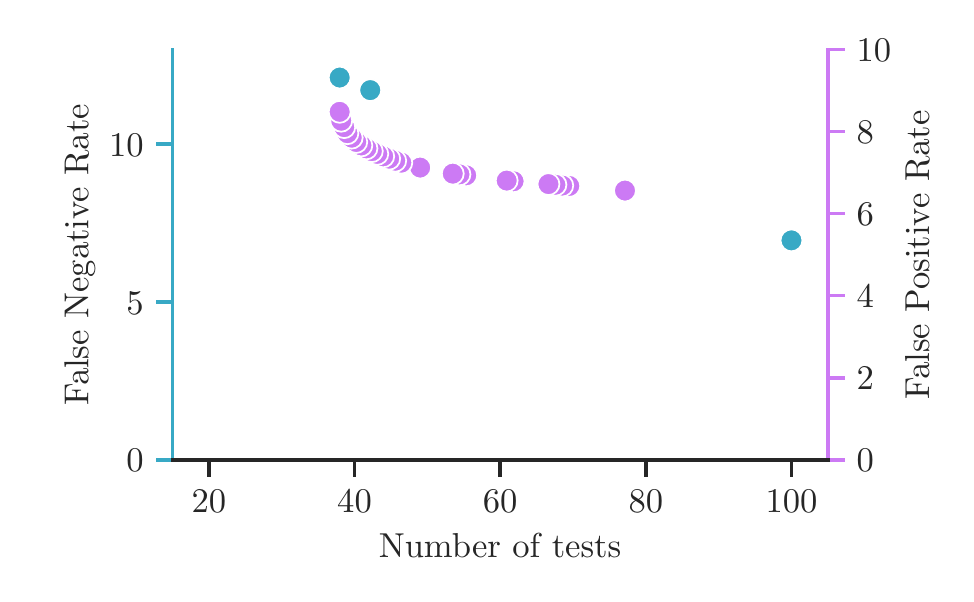}}
\subfloat[$s_e=s_p=0.85$]{\includegraphics[width=0.32\linewidth]{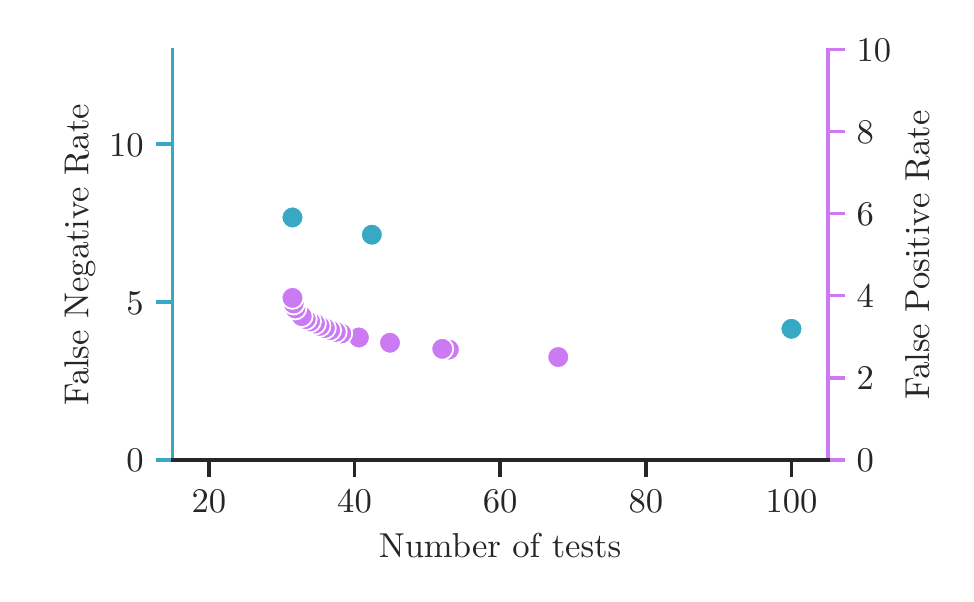}}
\subfloat[$s_e=s_p=0.95$]{\includegraphics[width=0.32\linewidth]{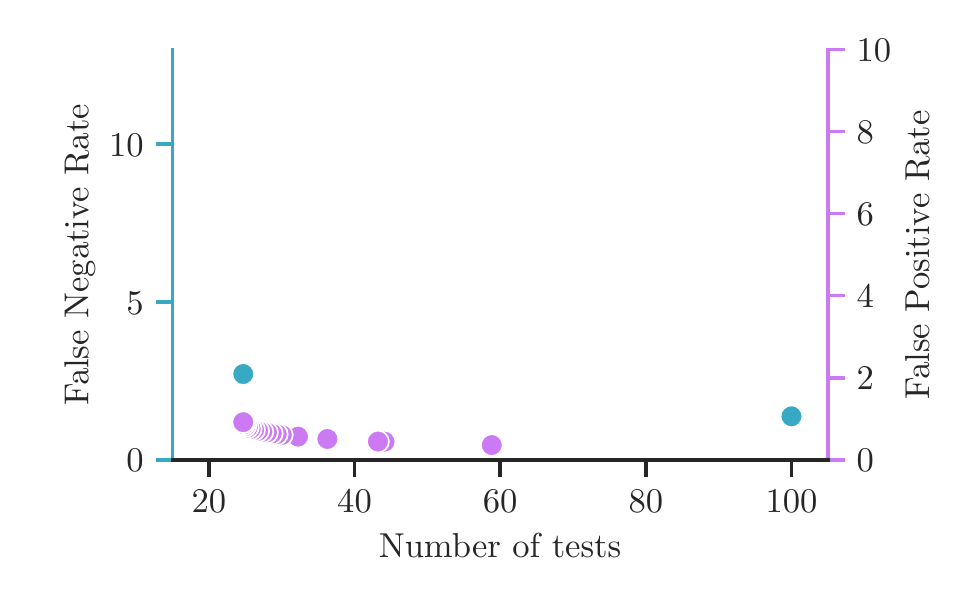}}
\caption{Average number of tests, false negative rate and false positive rate achieved by our method (Algorithm~\ref{alg:dp}) under different values of the parameters $\lambda_1$ and $\lambda_2$ and 
different levels of specificity $s_e$ and sensitivity $s_p$.
In each panel, we either penalize the false negative rate (i.e., we vary $\lambda_1$ and set $\lambda_2 = 0$) or the false positive rate (i.e., we vary $\lambda_2$ and set $\lambda_1 = 0$).
Accordingly, for the former, we show the false negative rate vs average number of tests (in blue) and, for
the latter, we show the false positive rate vs average number of tests (in pink).
%
%
Here, we set the number of contacts to $N=100$ and sample the number of positive infections from a truncated negative binomial distribution with mean $r = 2.5$ and dispersion parameter $k = 0.1$.
In each experiment, we estimate averages using $100{,}000$ samples.} \label{fig:performance-vs-lambdas}
\end{figure}

%
%
%
%
%

Next, we investigate to what extent our method offers a competitive advantage with respect to Dorfman'{}s 
method for additional values of the reproductive number $r$ and dispersion parameter $k$
different than those estimated during the early phase of the COVID-19 pandemic.
Figure~\ref{fig:performance-vs-R-k} summarizes the results, which show that our method offers the 
greatest competitive advantage whenever the number of secondary infections is overdispersed, i.e., 
$k \rightarrow 0$. 
%
%
Moreover, as the number of contacts $N$ increases, the competitive advantage is greater for larger
values of the reproductive number $r$. 
%

%
%
%
%
%


Finally, we explore the trade-off between the average number of tests that our method achieves and the false positive and negative rates, under different values of the parameters $\lambda_1$ and $\lambda_2$.
Figure~\ref{fig:performance-vs-lambdas} summarizes the results, which show that, to achieve lower false negative and false positive rates, more tests need to be performed.
When trading off the number of tests with the number of false positives ($\lambda_1=0$, $\lambda_2 > 0$), our method gradually changes the average pool size, leading to many possible trade-off points between 
the number of tests and the false positive rate.
When $\lambda_2$ takes small values, the optimal solution leads to pool sizes that minimize the number of tests, while the solution consists of pools of two contacts when $\lambda_2$ gets significantly larger.
When balancing the number of tests with the number of false negatives ($\lambda_1>0$, $\lambda_2 = 0$), there are only two or three possible solutions, with the extreme ones corresponding to pool sizes minimizing the number of tests and pools of size one, \ie, individual testing for all.
We noticed that the expected number of false negatives in a pool grows linearly with its size for sizes greater than one, therefore, making the exact size of the pool irrelevant in terms of the expected total number of false negatives.
As a consequence, this leads to only a few optimal solutions where some of the contacts are individually tested while the 
rest of them are split into pools.
\looseness=-1

Our results have direct implications for the allocation of limited and imperfect testing resources in future 
pandemics whenever there exists evidence of substantial overdispersion in the number of secondary infections. 
In this context, we acknowledge that more research is needed to more accurately characterize the level of 
overdispersion in a pandemic.
Moreover, it would be interesting to extend our algorithm using distributions other that the negative 
binomial and practically evaluate our method in randomized control studies.
%

%
%
%
%
%
%


\section*{Acknowledgements}
We would like to thank Vipul Bajaj for fruitful discussions.

{
\bibliographystyle{unsrt}
\bibliography{refs}
}




\end{document}